\documentclass[a4paper,12pt]{amsart}

\usepackage{amsmath,amssymb,amsthm,graphicx,stmaryrd}
\usepackage{verbatim,enumitem,xfrac,youngtab,hyperref}

\newcommand{\CC}{\mathbb{C}}
\newcommand{\EE}{\mathbb{E}}

\newcommand{\NN}{\mathbb{N}}
\newcommand{\PP}{\mathbb{P}} 
\newcommand{\QQ}{\mathbb{Q}} 
\newcommand{\RR}{\mathbb{R}}

\newcommand{\ZZ}{\mathbb{Z}}

\newcommand{\cC}{\mathcal{C}}

\newcommand{\cM}{\mathcal{M}} 
\newcommand{\cH}{\mathcal{H}}

\newcommand{\cT}{\mathcal{T}}
\newcommand{\cS}{\mathcal{S}}

\renewcommand{\a}{\alpha}
\newcommand{\D}{\Delta} 
\renewcommand{\d}{\delta}

\newcommand{\g}{\gamma}

\newcommand{\la}{\lambda}
\renewcommand{\b}{\beta} 
\renewcommand{\k}{\kappa} 

\newcommand{\om}{\omega} 
 
\newcommand{\s}{\sigma}
  
\newcommand{\eps}{\varepsilon}

\newcommand{\el}{\langle} 
\newcommand{\er}{\rangle}
\newcommand{\tr}{\mathrm{tr}}

\newcommand{\lra}{\leftrightarrow}

\renewcommand{\b}{\beta}
\newcommand{\oo}{\infty}

\newcommand{\ket}[1]{|#1\rangle}
\newcommand{\pet}[1]{|#1)}

\newcommand{\sm}{\setminus}

\newcommand{\se}{\subseteq}

\newcommand{\ol}{\overline}

\newcommand{\crit}{\mathrm{c}}

\newcommand{\spi}{S^1}
\newcommand{\spii}{S^2}
\newcommand{\spiii}{S^3}
\newcommand{\spj}{S^j}
\newcommand{\bS}{\mathbf{S}}
\newcommand{\bJ}{\mathbf{J}}
\newcommand{\bK}{\mathbf{K}}

\newcommand{\one}{\hbox{\rm 1\kern-.27em I}}

\newcommand{\dom}{\trianglerighteq}

\newcommand{\tsum}{\textstyle\sum}

\newcommand{\be}{\begin{equation}}
\newcommand{\ee}{\end{equation}}

\newtheoremstyle{slthm}
     {}
     {\baselineskip}
     {\slshape}
     {\parindent}
     {\scshape}
     {.}
     { }
     {}

\theoremstyle{slthm}
\newtheorem{definition}{Definition}[section]
\newtheorem{theorem}[definition]{Theorem}
\newtheorem{proposition}[definition]{Proposition}
\newtheorem{lemma}[definition]{Lemma}

\title[Quantum spin systems and interchange processes]
{The free energy in a class of quantum spin systems 
and interchange processes}
\author{J. E. Bj\"ornberg}
\address{Department of Mathematics, University of Copenhagen, Denmark}
\date{\today}

\begin{document}

\begin{abstract}
We study a class of quantum spin systems in the
mean-field setting of the complete graph.  
For spin $S=\tfrac12$
the model is the Heisenberg ferromagnet, for general 
spin $S\in\tfrac12\NN$ it has
a probabilistic representation as a cycle-weighted 
interchange process.
We determine the free energy and 
the critical temperature (recovering results by T\'oth
and by Penrose when $S=\tfrac12$).
The critical temperature is shown to
coincide (as a function of $S$)
with that of the $q=2S+1$ state
classical Potts  model, and the phase transition is 
discontinuous when $S\geq1$.
\end{abstract}

\maketitle

\section{Introduction}

It has been well-known since the work of 
T\'oth~\cite{toth-heis} and 
Aizenman and Nachtergaele~\cite{an} in the early
1990's that many quantum spin-systems can be analyzed using
probabilistic representations.  
T\'oth's representation of the (spin-$\tfrac12$) Heisenberg
ferromagnet in terms of random transpositions
is particularly appealing in its simplicity. 
However, though simple to define, it
has proved challenging to obtain rigorous results using this
representation.  
While substantial progress has
been made on several other models using probabilistic 
representations~\cite{B-irb,B-van,BG,bjo-uel,CI,CNS,lees,uel-jmp},
proving a phase-transition in
the ferromagnetic Heisenberg model on the lattice $\ZZ^d$
remains an open challenge.
 
For mean-field variants there has been more progress, and
related models have recently received quite a lot of
attention in the probability 
literature~\cite{alon-kozma,angel,berestycki,berestycki-kozma,
bjo-cycles,hammond-sharp,kmu,schramm}.  
The free energy of the spin-$\tfrac12$ 
Heisenberg ferromagnet on the complete graph was
determined already in 1990:
by T\'oth~\cite{toth-bec}  using a random-walk
representation, and simultaneously but independently by 
Penrose~\cite{penrose} 
by explicitly diagonalizing the Hamiltonian.

Here we extend the latter 
results to a class of spin $S\in\tfrac12\NN$ 
models, with Hamiltonian equal to a sum of transposition-operators
(see below for a precise definition).  Probabilistically, the model
naturally generalizes T\'oth's permutation-representation:
a weight factor $2^{\#\mathrm{cycles}}$ is replaced by 
$(2S+1)^{\#\mathrm{cycles}}$.   Our approach is different both from
that of T\'oth and that of Penrose.
The key step is to
obtain an expression for the partition function in terms of
the irreducible representations of the symmetric group.
Perhaps our most surprising 
result is a connection to the classical Potts-model:
we show that the critical temperature
of our model, as a function of $S$, coincides with that of the
$q=2S+1$ state Potts model.

We now define the model and state
our primary results.

\subsection{Model and main results}

We let $\spi,\spii,\spiii$ denote the usual spin-operators, satisfying
the relations
\[
[\spi,\spii]=i\spiii,\quad
[\spii,\spiii]=i\spi,\quad
[\spiii,\spi]=i\spii,
\]
where $i=\sqrt{-1}$.
For each $S\in\tfrac12\NN$ we work with the standard spin-$S$
representation, where the $\spj$ are Hermitian matrices acting on
$\cH=\CC^{2S+1}$.  We fix an orthonormal basis for $\cH$ 
consisting  of eigenvectors for $\spiii$,  denoting the basis
vector with eigenvalue $a\in\{-S,-S+1,\dotsc,S\}$ by 
$\ket a$.

Let $G=K_n=(V,E)$ be the complete graph on $n$ vertices, i.e.\ the
graph with vertex set $V=\{1,\dotsc,n\}$ and edge set $E=\binom{V}{2}$
consisting of one edge (bond) per pair $x\neq y$ of vertices.  For
each $x\in V$ we take a copy $\cH_x$ of $\cH$, and 
we form the tensor product $\cH_V=\otimes_{x\in V}\cH_x$.
An orthonormal basis for $\cH_V$ is given by the vectors
$\ket{\mathbf{a}}=\otimes_{x\in V}\ket{a_x}$ for 
$\mathbf{a}=(a_x)_{x\in V}\in\{-S,\dotsc,S\}^V$.
If $A$ is an operator acting on $\cH$ we define $A_x$ acting on
$\cH_V$ by $A_x=A\otimes\mathrm{Id}_{V\sm \{x\}}$.

The transposition operator $T_{xy}$ 
on $\cH_V$ is defined as follows.
For each pair $x\neq y$ of vertices,  $T_{xy}$
is given by its action on the basis elements $\ket{\mathbf{a}}$:  
\be\label{T-def}
T_{xy}\otimes_{z\in V}\ket{a_z}=
\otimes_{z\in V} \ket{a_{\tau(z)}},
\ee
where $\tau=(x,y)$ is the transposition of $x$ and $y$:
\[
\tau(z)=\left\{
\begin{array}{ll}
y, & \mbox{if } z=x,\\
x, & \mbox{if } z=y,\\
z, & \mbox{otherwise.}
\end{array}
\right.
\]
Thus $T_{xy}$ interchanges the $x$ and $y$ entries of 
$\ket{\mathbf{a}}$.

Our model has the Hamiltonian
\be
H=H_n=-\sum_{xy\in E} (T_{xy}-1)
\ee
acting on $\cH_V$.  We take the inverse-temperature of the form $\b/n$
for constant $\b>0$, thus the partition function is
\be
Z_n(\b)=\tr (e^{-(\sfrac\b n) H_n}).
\ee
We note that $T_{xy}$ may be expressed as a polynomial in the
operators $\bS_x\cdot\bS_y=\sum_{j=1}^3\spj_x\spj_y$.  
For example, when $S=\tfrac12$ we have
that $T_{xy}=2(\bS_x\cdot\bS_y)+\tfrac12$, and when $S=1$ that
$T_{xy}=(\bS_x\cdot\bS_y)^2+(\bS_x\cdot\bS_y)-1$.  (See
Proposition~\ref{T-prop} 
in the appendix for the general case.)
Thus for $S=\tfrac12$ we recover the Heisenberg ferromagnet at
inverse-temperature $2\b/n$.

Our first main result is an explicit formula 
for the free energy.  
For each $S\in\tfrac12\NN$, let $\theta=2S+1$ and let
\[
\D=\D_\theta=\{x=(x_1,\dotsc,x_\theta)\in[0,1]^\theta:
x_1\geq\dotsb\geq x_\theta,\tsum_{j=1}^\theta x_j=1\}.
\]
Define the function $\phi_\b:\D\to\RR$ by
\be\label{phi-def}
\phi_\b(x)=\frac\b2\Big(\sum_{j=1}^\theta x_j^2-1\Big)
-\sum_{j=1}^\theta x_j \log x_j.
\ee

\begin{theorem}\label{free-en-thm}
We have that
\be
\lim_{n\to\oo} \frac1n \log Z_n(\b)=
\max_{x\in\D} \phi_\b(x).
\ee
\end{theorem}

As mentioned previously, our analysis relies on a probabilistic
representation.  
We describe this now.
Let $\PP(\cdot)$ be a probability measure governing a
collection $\om=(\om_{xy}:xy\in E)$ of independent rate 1 Poisson
processes on $[0,\b/n]$, indexed by the edges of $G$.  Thus each
$\om_{xy}$ is a random (almost-surely finite) subset of $[0,\b/n]$;
the number of elements of $\om_{xy}$ in an interval $[s,t]$
has the Poisson distribution Po($t-s$), and these numbers are
independent for disjoint intervals. 
We think of $[0,\b/n]$ as a time-interval.
See Figure~\ref{loops-fig} for a pictorial representation.

 As explained in e.g.~\cite[eq.~(2.11)]{an}
we have from the Lie--Trotter product formula that
\be\label{trotter}
e^{-(\sfrac\b n) H_n}=\EE\Big[
\sideset{}{^{\,\star}}\prod_{(xy,t)\in\om} T_{xy}\Big],
\ee
where $\Pi^\star$ is the time-ordered product over all elements of
$\om$.  
In light of~\eqref{T-def} and~\eqref{trotter} we may think of each
point $(xy,t)\in\om$ as representing a transposition of $x,y\in\{1,\dotsc,n\}$
at time $t$.  We let $\s=\s(\om)=\prod^\star_{(xy,t)\in\om}(x,y)$ 
denote the (time-ordered) composition
of these transpositions from time 0 to time $\b/n$.  
Thus $\s\in \cS_n$, the symmetric group on $n$
letters.  

Recall that each $\s\in \cS_n$ may be written as a product of disjoint
cycles (orbits).  Let $\ell=\ell(\om)$ denote the number of such
cycles of $\s(\om)$, including singletons.  Taking the trace
in~\eqref{trotter} we find that we get a contribution of 1 from each
basis vector $\ket{\mathbf{a}}$ for which the function
$\mathbf{a}:V\to\{-S,\dotsc,S\}$ is constant on each cycle of
$\s(\om)$.  (Figure~\ref{loops-fig} is helpful in verifying this statement.)  
From the other $\ket{\mathbf{a}}$ we get contribution 0.  
Writing
$\theta=2S+1$, as before, 
for the number of possibilities per cycle, we
conclude that 
\be\label{pf-exp}
Z_n(\b)=\tr(e^{-(\sfrac\b n) H_n})=\EE[\theta^{\ell(\om)}].
\ee

\begin{figure}
\centering
\includegraphics[scale=.8]{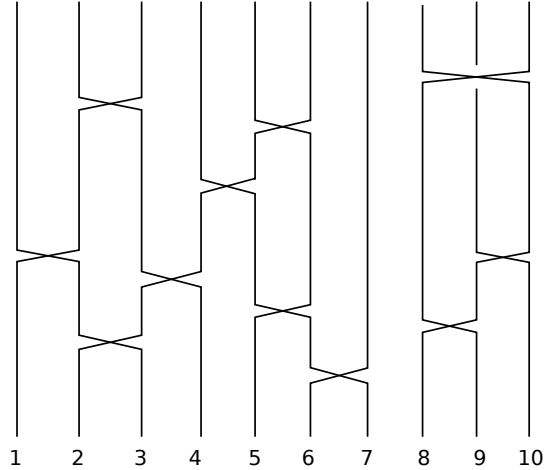}
\caption{
A sample $\om$, with
the vertex set $V=\{1,\dotsc,10\}$ on the 
horizontal axis and time going upwards.
Elements $(xy,t)\in\om$ are represented 
as crosses, and are to be thought of as
transpositions.  (In this picture, for clarity only, most crosses
occur between consecutive vertices.)
Here $\s(\om)=(1,3)(2,6,7,4)(9,10)(5)(8)$
and $\ell(\om)=5$.
}
\label{loops-fig}
\end{figure}

In order to identify a phase-transition we will work also with a
`weighted' version of~\eqref{pf-exp}.  Let $\cC=\cC(\om)$ denote the
set of cycles of the permutation $\s(\om)$, and for $h\in\RR$ write 
\be\label{pf-exp-2}
Z_n(\b,h)=
e^{-(\sfrac h\theta) n}
\EE\Big[\prod_{\g\in\cC} (e^{h|\g|}+\theta-1)\Big],
\ee
where $|\g|$ denotes the size of the cycle $\g$.  Note that
$Z_n(\b,0)=Z_n(\b)$.  
We will later (see Theorem~\ref{free-energy-thm}) 
obtain an explicit expression for
the limit
$z(\b,h)=\lim_{n\to\oo} \frac1n\log Z_n(\b,h)$.
(Theorem~\ref{free-en-thm} is the special case $h=0$ of that result.)
Our second main result concerns the right derivative 
$z^+(\b)=\lim_{h\downarrow 0}\frac{z(\b,h)-z(\b,0)}{h}$ 
of $z(\b,h)$ at $h=0$. 

\begin{theorem}\label{mag-thm}
Define
\be\label{b-crit}
\beta_\crit(\theta)=\left\{
\begin{array}{ll}
2, & \mbox{if } \theta=2,\\
2\big(\tfrac{\theta-1}{\theta-2}\big)\log(\theta-1), &
\mbox{if } \theta\geq3.
\end{array}\right.
\ee
Then for all $\theta\in\{2,3,\dotsc\}$ we have that
\be
z^+(\b)\left\{\begin{array}{ll}
=0, & \mbox{if } \b<\b_\crit,\mbox{ or } \theta=2\mbox{ and }
          \b=\b_\crit,\\
>0, & \mbox{if } \b>\b_\crit,\mbox{ or } \theta\geq3\mbox{ and }
          \b=\b_\crit.
\end{array}\right.
\ee
\end{theorem}
Thus, the critical inverse-temperature is given by~\eqref{b-crit}, and
the phase-transition is continous for $\theta=2$ (i.e.\ $S=\tfrac12$)
and discontinuous for $\theta\geq3$ (i.e.\ $S\geq1$).
We reiterate that the case $\theta=2$ was fully understood
previously~\cite{penrose,toth-bec}.

\subsection{Discussion}

Theorem~\ref{mag-thm} has
 consequences for the following
\emph{weighted interchange process}.  Recall the measure $\PP$
governing the random permutation $\s(\om)$, obtained as the
composition of a  process of transpositions up to time
$\b/n$.  For each $\theta>0$ one can define
another probability measure $\PP_\theta$ by requiring 
$\frac{d\PP_\theta}{d\PP}\propto\theta^{\ell(\om)}$.  The measure
$\PP_\theta$ allows for probabilistic interpretation
of correlation functions.  For example when $S=\tfrac12$:
\[
\el \spiii_x\spiii_y\er=\tfrac12\PP_2(x\lra y),
\]
where $\{x\lra y\}$ is the event that $x$ and $y$ belong to the same
cycle.  Similar relations hold for other $\theta$.
Magnetic ordering is thus accompanied by the occurrence of
large cycles in a $\PP_\theta$-distributed random permutation.

For each $k\geq0$ 
let $X_n(k)=\tfrac1n\sum_{|\g|\geq k}|\g|$ denote the
fraction of vertices in cycles of size at least $k$ in the random
permutation $\s(\om)$.
From Theorem~\ref{mag-thm} we will deduce the
following: 
\begin{proposition}\label{X-prop} 
If $\theta\in\{2,3,\dotsc\}$ and $z^+(\b)=0$ then for any sequence
$k=k_n\to\oo$ and any fixed $\eps>0$, there is a $c>0$ such that
\[
\PP_\theta(X_n(k)\geq\eps)\leq e^{-cn}.
\]
\end{proposition}
T\'oth's formula~\cite[eq.~(5.2)]{toth-heis}
suggests that a converse to Proposition~\ref{X-prop}
should also hold, i.e.\ that there are cycles of size of the order $n$
when $z^+(\b)>0$.  We have not been able to prove this.
Note, however, that cycles of
order $n$ do occur whenever $\b>\theta\geq1$.
For $\theta=1$ this was proved by Schramm~\cite{schramm},
and for $\theta>1$ it was proved in~\cite{bjo-cycles}
using Schramm's result.

Theorem~\ref{mag-thm} also points to a connection to the classical Potts
model.  In that model, one considers random assignments
$\eta=(\eta_x:x\in V)$ of the values $1,2,\dotsc,q$ to the vertices
$x\in V$, for some fixed $q\in\{2,3,\dotsc\}$.  
Each such assignment receives probability proportional to 
\[
\exp\big(\tfrac\b n\tsum_{xy\in E}\d_{\eta_x,\eta_y}\big).
\]
It was proved by Bollob\'as, Grimmett and Janson 
in~\cite{bgj} (in the more general context of the
random-cluster-representation) that a phase-transition
occurs in this model
at the point $\b=\b_\crit(q)$ with $\beta_\crit(\cdot)$ as given
in~\eqref{b-crit}.  
This equality of critical points may indicate a deeper connection
between the two models, which we hope to explore in future work.

\subsection{Outline}
Over the following three sections we will prove somewhat
more detailed versions of Theorems~\ref{free-en-thm}
and~\ref{mag-thm} and Proposition~\ref{X-prop}.  
In Section~\ref{char-sec} we first 
obtain a formula for $Z_n(\b,h)$ for finite $n$,
stated in Lemma~\ref{pf-lem}.
This formula is amenable to asymptotic analysis,
which we perform in Section~\ref{lim-sec}.  The main result
of that Section is Theorem~\ref{free-energy-thm},
where we compute
$\lim_{n\to\oo}\tfrac1n\log Z_n(\b,h)$.
In Section~\ref{pt-sec} we use the latter result
to describe the phase transition and identify
the critical point.
Some additional proofs are given in the Appendix.

\section{Character decomposition of the partition function}
\label{char-sec}

In this section we obtain an expression for the
partition function $Z_n(\b,h)$ in terms of the irreducible
characters of the symmetric group.
From now on we will usually only refer to the spin 
$S\in\tfrac12\NN$ via the parameter 
$\theta=2S+1\in\{2,3,\dotsc\}$.  Recall that 
$\s=\s(\om)\in\cS_n$ is the random permutation introduced
below~\eqref{trotter}, that $\cC=\cC(\om)$ is the set of
cycles in a disjoint-cycle decomposition of $\s$, 
and that $\ell=\ell(\om)=|\cC(\om)|$ is the number
of cycles.

By a \emph{composition} $\k$ of $n$ 
we mean a vector $\k=(\k_1,\dotsc,\k_\theta)$
with non-negative integer entries, such that
$\sum_{j=1}^\theta\k_j=n$.  Note that we restrict the
number of entries to be exactly $\theta$, and
that we allow some $\k_j$
to be $=0$.  A composition $\la$
 is called a \emph{partition}  if in addition 
$\la_j\geq\la_{j+1}$ for all $j$,
in which case  we write $\la\vdash n$.  
 Any composition may be rearranged to
form a partition.   Given a partition $\la$, let $K(\la)$ denote the set
of compositions that can be obtained by re-ordering the entries of
$\la$.  Clearly $1\leq |K(\la)|\leq\theta!$.
We write $\binom{n}{\la}$ for the multinomial coefficient
\[
\binom{n}{\la}=\frac{n!}{\la_1!\la_2!\dotsb\la_\theta!}.
\]

\subsection{Colouring-lemma}

Let $p_1,\dotsc,p_\theta$ be probabilities, i.e.\ 
non-negative numbers summing to 1.
Write $f(\s)=\PP(\s(\om)=\s)$ for the 
distribution function of $\s(\om)$.  Note that $f(\cdot)$ is a
class-function, i.e.\ $f(\s)=f(\pi)$ whenever $\s$ and $\pi$ have the
same cycle-type.  (This uses that we are working on the complete
graph.) 
For $\la\vdash n$
we write $\cT_\la$ for the \emph{Young subgroup} of 
$\cS_n$, i.e.\ the subgroup consisting of those permutations which fix
each of the sets 
\[
\{1,\dotsc,\la_1\},\quad
\{\la_1+1,\dotsc,\la_1+\la_2\},\quad
\mbox{etc}.
\]

\begin{lemma}[Colouring-lemma]\label{col-lem}
We have that
\[
\EE\Big[\prod_{\g\in\cC}
\Big(\sum_{i=1}^\theta p_i^{|\g|}\Big)\Big]
=\sum_{\la\vdash n} \binom{n}{\la}
\Big(\sum_{\k\in K(\la)} \prod_{i=1}^\theta p_i^{\k_{i}} \Big)
\sum_{\s\in \cT_\la} f(\s).
\]
\end{lemma}

\begin{proof}
Colour each vertex of 
$V=\{1,\dotsc,n\}$ independently using the colours
$1,\dotsc,\theta$, 
colour $\# i$ with probability $p_i$.
Write $\cM$ for the event that all cycles of $\s$ are
monochromatic.   The conditional probability
of $\cM$ given $\s$ is 
\[
\prod_{\g\in\cC}
\Big(\sum_{i=1}^\theta p_i^{|\g|}\Big),
\]
so the left-hand-side of the claim is just $\PP(\cM)$.
On the other hand,  by assigning the colours first 
we see that
\be\label{cs-eq}
\PP(\cM)=\sum_{C_1,\dotsc,C_\theta}
\Big(\prod_{i=1}^\theta p_i^{|C_i|}\Big)
\PP\big(\s\in \cT(C_1,\dotsc,C_\theta)\big)
\ee
where the sum is over all (ordered) set partitions
$C_1,\dotsc,C_\theta$ of $\{1,\dotsc,n\}$, and
$\cT(C_1,\dotsc,C_\theta)$ is the subgroup of $\cS_n$ consisting of
permutations which fix each of the sets $C_i$.

Let $\la\vdash n$ be the partition of $n$
obtained by ordering the $|C_i|$ by size.  
Then there is some $\pi\in \cS_n$ such that 
\be
\pi^{-1} \cT(C_1,\dotsc,C_\theta) \pi= \cT_\la.
\ee
Indeed, conjugation corresponds to relabelling
the vertices, so we simply choose the appropriate relabelling of the
sets $C_i$.  It follows that
\[\begin{split}
\PP\big(\s\in \cT(C_1,\dotsc,C_\theta)\big)
&=\sum_{\s\in \cT(C_1,\dotsc,C_\theta)} f(\s)
=\sum_{\s\in \cT_\la} f(\pi\s \pi^{-1})\\
&=\sum_{\s\in \cT_\la} f(\s),
\end{split}\]
since $f(\cdot)$ is a class-function.
Putting this into~\eqref{cs-eq} and summing over all possible
$\la\vdash n$ we get that 
\be
\PP(\cM)=
\sum_{\la\vdash n}  
\Big(\sum_{\s\in \cT_\la} f(\s)\Big)
\Big(\sum_{C_1,\dotsc,C_\theta}
\prod_{i=1}^\theta p_i^{|C_i|}\Big)
\ee
where now the sum over the $C_i$ is restricted to those
with the property that $(|C_1|,\dotsc,|C_\theta|)\in K(\la)$.   
This sum may be performed by first summing over all
$\k\in K(\la)$,
and then over all choices of the sets $C_i$ with $\k_i=|C_i|$.
For each fixed $\k$, there are $\binom{n}{\la}$ choices of
the sets.   It follows that 
\[
\sum_{C_1,\dotsc,C_\theta}
\prod_{i=1}^\theta p_i^{|C_i|}
=\binom{n}{\la} \sum_{\k\in K(\la)} \prod_{i=1}^\theta p_i^{\k_i},
\]
which proves the claim.
\end{proof}

Introduce the notation
\be\label{G-def}
G_n(\la)=\binom{n}{\la}\PP(\s\in \cT_\la)=
\binom{n}{\la} \sum_{\s\in\cT_\la}f(\s),\quad
\mbox{for } \la\vdash n.
\ee
Taking all the $p_i=\tfrac1\theta$ in Lemma~\ref{col-lem}
and using~\eqref{pf-exp}
we get (cancelling a factor $\theta^{-n}$) that
$Z_n(\b)=\EE[\theta^{\ell(\om)}]=\sum_{\la\vdash n} |K(\la)| G_n(\la)$.
More generally, we may take
\be
p_1=p e^h, \qquad p_2=\dotsb=p_\theta=p,
\qquad\mbox{for }h\in\RR,
\ee
with appropriate normalization $p=(e^h +\theta -1)^{-1}$.
Lemma~\ref{col-lem} and~\eqref{pf-exp-2} give that
\be
e^{(\sfrac h \theta) n} Z_n(\b,h)=
\EE\Big[\prod_{\g\in\cC} (e^{h|\g|}+\theta -1)\Big]
=\sum_{\la\vdash n}
\Big(\sum_{\k\in K(\la)} e^{h \k_1}\Big)G_n(\la).
\ee
The factors $\sum_{\k}e^{h\k_1}$ are bounded by simple
expressions.  Indeed,  if $h\geq0$ then,
since $e^{h\la_1}$ is a summand in the sum over $K(\la)$,
we have that 
\be
\sum_{\k\in K(\la)} e^{h \k_1}=
e^{h \la_1} \sum_{\k\in K(\la)} e^{h (\k_1-\la_1)}
\left\{
\begin{array}{l}
\geq e^{h \la_1}\\
\leq  \theta! e^{h \la_1},
\end{array}\right.
\ee
since $\la_1$ is the largest of the $\k_i$.
Similarly, if $h\leq0$ then
\be
\sum_{\k\in K(\la)} e^{h \k_1}
\left\{
\begin{array}{l}
\geq e^{h \la_\theta}\\
\leq  \theta! e^{h \la_\theta}.
\end{array}\right.
\ee

We will use the notation $f(n)\asymp g(n)$ to denote that there is a
constant $C>0$ such that $\tfrac1C g(n)\leq f(n)\leq C g(n)$ for all 
$n$.  We may summarize the above as follows:
\begin{lemma}\label{pf-lem}
With $G_n(\la)=\binom{n}{\la}\PP(\s\in\cT_\la)$ as in~\eqref{G-def},
we have that
\[\begin{split}
Z_n(\b)&\asymp \sum_{\la\vdash n} G_n(\la),\quad\mbox{and}\\
e^{(\sfrac h\theta) n}Z_n(\b,h)&\asymp 
\left\{\begin{array}{ll}
\sum_{\la\vdash n}e^{h\la_1} G_n(\la), &
\mbox{if } h> 0,\\
\sum_{\la\vdash n}e^{h\la_\theta} G_n(\la), &
\mbox{if } h< 0.
\end{array}
\right.
\end{split}\]
\end{lemma}

\subsection{Some representation theory}

From Lemma~\ref{pf-lem} it is clear that the probabilities
$\PP(\s\in \cT_\la)$ are important.  We now express them using the
irreducible representations of $\cS_n$.
For background on the representation theory of $\cS_n$
we refer to e.g. Fulton--Harris~\cite{fulton-harris}.

 The irreducible representations of $\cS_n$
are indexed by partitions $\mu\vdash n$.    
(In this description we temporarily omit 
our convention that partitions
have at most $\theta$ non-zero parts.)
It is convenient to
represent $\mu\vdash n$ by its Young-diagram, 
as in Figure~\ref{ytab-fig}.
We write
$U_\mu$ for the irreducible representation corresponding to 
$\mu\vdash n$, and
$\chi_\mu$ for its character.  Let $V_\la$ denote the
coset representation of the subgroup $\cT_\la$,
that is $V_\la$ is a vector space spanned
by the cosets $\pi \cT_\la$ and $\cS_n$ acts by
left multiplication.  By Young's
rule~\cite[Corollary~4.39]{fulton-harris},  
the representation $V_\la$ decomposes as a
direct sum of irreducible representations with known multiplicities:
\be\label{V-decomp}
V_\la=\bigoplus_{\mu\vdash n}K_{\mu\la} U_\mu.
\ee
Here the multiplicities $K_{\mu\la}$ are
the \emph{Kostka numbers}: 
$K_{\mu\la}$ equals the number of ways to fill the Young
diagram for $\mu$ with $\la_1$ 1's, $\la_2$ 2's
et.c.\ so that the rows are weakly increasing
and the columns are strictly increasing.
See Figure~\ref{ytab-fig} again.
\begin{figure}
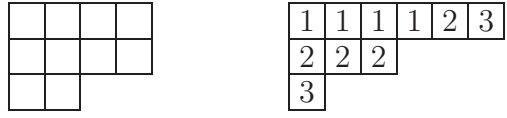

\centering
$\yng(4,4,2)$ \qquad\qquad 
$\young(111123,222,3)$
\caption{Left:  Young diagram of the partition 
$\la=(4,4,2)\vdash 10$.  Right:  diagram for
$\mu=(6,3,1)\dom\la$ filled with $\la_1=4$ 1's,
$\la_2=4$ 2's and $\la_3=2$ 3's so that rows are weakly
increasing and columns strictly increasing.}
\label{ytab-fig}
\end{figure}

We say that $\mu$ \emph{dominates} $\la$,
written $\mu\dom\la$, if for each $i$ we have
that $\mu_1+\dotsb+\mu_i\geq \la_1+\dotsb+\la_i$.
Note that $K_{\mu\la}=0$ unless $\mu\dom\la$.
In particular, if $\la$ has at most $\theta$ non-zero parts, then
$K_{\mu\la}=0$ unless $\mu$ also has at most $\theta$ non-zero parts. 
 Writing $\psi_\la$ for the character of $V_\la$ it follows
from~\eqref{V-decomp} that
\be\label{char-decomp-V}
\psi_\la=\sum_{\mu\vdash n} K_{\mu\la} \chi_\mu.
\ee
Let $\el\cdot,\cdot\er$ denote the inner product
of functions on $\cS_n$ given by
\[
\el f,g\er=\frac{1}{n!}\sum_{\s\in \cS_n}f(\s)\overline{g(\s)}.
\]
Lemma 1 in Alon--Kozma~\cite{alon-kozma} tells us that for 
a class function $f$ we have
\be
\sum_{\s\in \cT_\la} f(\s)=|\cT_\la|\el f,\psi_\la\er
\ee
so using~\eqref{char-decomp-V} we see that
\be\label{ak}
\sum_{\s\in \cT_\la} f(\s)=|\cT_\la|
\sum_{\mu\vdash n} K_{\mu\la} \el f,\chi_\mu\er.
\ee
Now let $f(\s)=\PP(\s(\om)=\s)$ as before.  
As already noted, this is a class-function.
The calculations in Lemma~1
of Berestycki--Kozma~\cite{berestycki-kozma}  show that
\be\label{bk}
\el f,\chi_\mu\er=
\tfrac{1}{n!} \tr(\hat f(\mu))
=\tfrac{1}{n!} d_\mu 
\exp\Big\{\frac\b n \binom{n}{2}[r(\mu)-1]\Big\}.
\ee
Here  $\hat f(\mu)$ denotes the
Fourier transform of $f$ at the irreducible
representation $U_\mu$, the number
$d_\mu$ is the dimension of $U_\mu$,
and finally $r(\mu)=\chi_\mu((1,2))/d_\mu$ 
is the character ratio at a transposition.
We note for future reference that
\be\label{r-eq}
\frac{\b}{n}\binom{n}{2}[r(\mu)-1]
=\frac\b2\Big[
\sum_{j=1}^\theta \frac{\mu_j(\mu_j-2j+1)}{n}
-(n-1)\Big]
\ee
see e.g.\ equation (7) in~\cite{berestycki-kozma}.

Putting together~\eqref{ak} and~\eqref{bk} gives
\be
\PP(\s\in \cT_\la)=\sum_{\s\in \cT_\la} f(\s)=
\frac{|T_\la|}{n!}\sum_{\mu\vdash n} 
d_\mu K_{\mu\la} \exp\Big\{t\binom{n}{2}[r(\mu)-1]\Big\}.
\ee
Noting that
$\frac{|\cT_\la|}{n!}=\binom{n}{\la}^{-1}$,
we obtain:
\begin{lemma}\label{G-lem}
\[
G_n(\la)=
\sum_{\mu\vdash n} 
d_\mu K_{\mu\la} \exp\Big\{\frac\b n\binom{n}{2}[r(\mu)-1]\Big\}.
\]
\end{lemma}

For $\theta=2$ the partitions $\mu$ can 
be indexed by the length of the second row,
and it is well-known (and easy to see) that
$K_{\mu\la}=1$ when $\mu\dom\la$.  In that case
Lemma~\ref{G-lem}
is essentially~\cite[eq.~(49)]{penrose}.

\section{Convergence-results}
\label{lim-sec}

In this section will use the expressions in Lemmas~\ref{pf-lem} 
and~\ref{G-lem} to identify the limit of $\frac1n\log Z_n(\b,h)$.

\subsection{Lemmas}

We first present  convergence-results in a slightly more
general form, which we will later apply to our specific problem.
Some of the arguments in this subsection 
are strongly inspired by~\cite[Section~6]{penrose} 
and~\cite[Section~3.4]{ruelle}.

Recall that
\[
\D=\{(x_1,\dotsc,x_\theta)\in [0,1]^\theta:
x_1\geq\dotsb\geq x_\theta, \tsum x_i=1\}.
\]
For $x,y\in\D$ we write $y\dom x$ if 
$y_1+\dotsb+y_i\geq x_1+\dotsb+x_i$
for all $i$.
For $x\in\D$ we write 
\be
\D(x)=\{y\in \D: y\dom x\}.
\ee
It is not hard to see that
$\D(\tfrac1\theta,\dotsc,\tfrac1\theta)=\D$.
Also note that each $\D(x)$, and hence also $\D$, is compact and
convex. 

Write $\|\cdot\|$ for the $\oo$-norm on $\RR^\theta$,
$\|x-y\|=\max_{i=1,\dotsc,\theta} |x_i-y_i|$.
Write $d_\mathrm{H}(\cdot,\cdot)$ for the associated Hausdorff
distance between  sets in $\RR^\theta$:
\[
d_\mathrm{H}(A,B)=\inf\{\eps\geq 0: 
A\se B^\eps\mbox{ and } B\se A^\eps\}
\]
where
$A^\eps=\{x\in\RR^\theta: \|x-a\|<\eps\mbox{ for some }a\in A\}$.

The proof of the following result is given in Appendix~\ref{dh-app}.
\begin{lemma}\label{D-lem}
Let $x,y\in\D$ with $\|x-y\|\leq \eps<\theta^{-2}$.   Then 
\[
d_\mathrm{H}(\D(x),\D(y))< \theta\eps^{1/2}.
\]
\end{lemma}

Now let $\phi:\D\to\RR$ be any continuous function
(we will later take  $\phi=\phi_\b$).
Since $\D$ is compact, $\phi$ is uniformly continuous.
Let $\phi^{(\la)}_n(\mu)$
be a sequence of functions of partitions $\la,\mu\vdash n$ 
converging \emph{uniformly} to $\phi$ in the
following sense:   there is a sequence
$\d_n\to 0$, not depending on $\la$ or $\mu$, such that
$|\phi^{(\la)}_n(\mu)-\phi(\mu/n)|\leq\d_n$ for all $n$.

\begin{lemma}\label{penrose-lem}
If $n\to\oo$ and $\la/n\to x\in\D$ then
\[
\frac1n\log\Big(\sum_{\mu\dom\la} 
\exp\big(n\,\phi^{(\la)}_n(\mu)\big)\Big)
\to \max_{y\in \D(x)} \phi(y).
\]
The maximum is attained since $\D(x)$ is compact and $\phi$
continuous. 
\end{lemma}
\begin{proof}
We first prove an upper bound.
Since the number of partitions of $n$ into at most $\theta$ parts is
at most $n^\theta$ we have that 
\[\begin{split}
\sum_{\mu\dom\la}\exp\big(n\, \phi^{(\la)}_n(\mu)\big)&
\leq n^\theta \max_{\mu\dom\la} \exp\big(n\, \phi^{(\la)}_n(\mu)\big)\\
&\leq n^\theta \exp\big(n\, \max_{\mu\dom\la} \phi(\mu/n)+n\d_n\big)
\end{split}\]
so that
\[
\frac1n\log \Big(\sum_{\mu\dom\la}\exp\big(n\, \phi^{(\la)}_n(\mu)\big)\Big)
\leq o(1)+ \max_{\mu\dom\la} \phi(\mu/n).
\]
Let $x_n=\la/n$, then $\mu\dom\la$ is equivalent to $\mu/n\in\D(x_n)$,
so we have that
\[
\max_{\mu\dom\la} \phi(\mu/n)\leq 
\max_{y\in\D(x_n)} \phi(y)=\phi(y_n^\star)
\]
for some $y_n^\star\in\D(x_n)$.  
Now we use Lemma~\ref{D-lem}:
given any $\d>0$ we have, for $n$ large
enough, that there is some $x^\star_n\in\D(x)$ such that
$\|x_n^\star-y_n^\star\|<\d$.  Since $\phi$ is uniformly continuous we
may, given $\eps>0$, pick $\d$ so that
$\|x_n^\star-y_n^\star\|<\d$ implies 
$|\phi(x_n^\star)-\phi(y_n^\star)|<\eps$.
Then
\[
\phi(y_n^\star) \leq \phi(x_n^\star)+\eps
\leq \max_{y\in\D(x)} \phi(y)+\eps,
\]
since $x^\star_n\in\D(x)$.  This shows that
\[
\frac1n\log \Big(\sum_{\mu\dom\la}\exp\big(n\, \phi^{(\la)}_n(\mu)\big)\Big)
\leq o(1)+ \max_{y\in\D(x)} \phi(y)+\eps,
\]
for any $\eps>0$, so
\[
\limsup_{n\to\oo, \la/n\to x}
\frac1n\log \Big(\sum_{\mu\dom\la}\exp\big(n\, \phi^{(\la)}_n(\mu)\big)\Big)
\leq \max_{y\in\D(x)} \phi(y).
\]

Now for the lower bound.  Pick some  $x^\star\in\D(x)$ where
$\phi$ attains its maximum over $\D(x)$.  As before, write
$x_n=\la/n$.  Using Lemma~\ref{D-lem} as before,  given $\d>0$ we have
that $\D(x_n)$ intersects the ball $B_\d(x^\star)$ of radius $\d$
around $x^\star$ provided that $n$ is large enough.  Write 
$\ol B_\d(x^\star)$ for the closed ball.  By the triangle inequality
we may further assume that $\D(x_n)\cap\ol B_\d(x^\star)$ contains
some point of the form $\mu/n$.  Thus
\be
\begin{split}
\sum_{\mu\dom\la}\exp\big( n\, \phi^{(\la)}_n(\mu)\big)
&\geq \sum_{\mu/n\in \D(x_n)}
\exp\big( n\, \phi(\mu/n)-n\d_n\big)\\
&\geq \min_{\mu/n\in \D(x_n)\cap\ol B_\d(x^\star)}
\exp\big( n\, \phi(\mu/n)-n\d_n\big)\\
&\geq \min_{y\in \D\cap\ol B_\d(x^\star)}
\exp\big( n\, \phi(y)-n\d_n\big)\\
&= \exp\Big( n \min_{y\in \D\cap\ol B_\d(x^\star)} \phi(y)-n\d_n\Big).
\end{split}
\ee
Hence
\[
\frac1n\log\Big(\sum_{\mu\dom\la}\exp\big( n\, \phi^{(\la)}_n(\mu)\big)\Big)
\geq \min_{y\in \D\cap\ol B_\d(x^\star)} \phi(y)-\d_n.
\]
By the uniform continuity of $\phi$, given $\eps>0$ we may pick $\d$
small enough such that 
\[
\min_{y\in \D\cap\ol B_\d(x^\star)} \phi(y)\geq \phi(x^\star)-\eps.
\]
This gives
\[
\limsup_{n\to\oo, \la/n\to x}
\frac1n\log\Big(\sum_{\mu\dom\la}\exp\big( n\, \phi^{(\la)}_n(\mu)\big)\Big)
\geq \phi(x^\star)-\eps,
\]
which proves the claim.
\end{proof}

The next result may be established 
straightforwardly using Lemma~\ref{D-lem}:

\begin{lemma}\label{g-lem}
The function  $g:\D\to\RR$ given by
$g(x)=\max_{y\dom x} \phi(y)$
is continuous.
\end{lemma}

We next present a slight extension of Lemma~\ref{penrose-lem}.
We assume that $\phi$ and 
$\phi^{(\la)}_n(\mu)$ are as before.
Write $y=(y_1,\dotsc,y_\theta)\in\RR^\theta$ and 
$y\cdot x=\sum y_ix_i$ for the usual  product.

\begin{lemma}\label{penrose-lem-2}
For any $y\in\RR^\theta$ we have as $n\to\oo$ that
\[
\frac1n\log\Big(\sum_{\la\vdash n} e^{y\cdot \la}
\sum_{\mu\dom\la}\exp\big( n\, \phi^{(\la)}_n(\mu)\big)\Big)
\to \max_{x\in\D}\big( y \cdot x + g(x)\big)
\]
where $g$ is the function in Lemma~\ref{g-lem}.
\end{lemma}  
\begin{proof}
Write $m(y)=\max_{x\in\D}\big( y \cdot x + g(x)\big)$.
Bounding the number of partitions by $n^\theta$ as before, we see
that 
\[
\begin{split}
\sum_{\la\vdash n} e^{y\cdot \la}
\sum_{\mu\dom\la}\exp\big( n\, \phi^{(\la)}_n(\mu)\big)
&\leq n^{2\theta} \cdot \max_{\la\vdash n} \max_{\mu\dom\la}
\exp\big( y\cdot\la+ n\, \phi(\mu/n)+n\d_n\big)\\
&\leq n^{2\theta} \cdot \max_{\la\vdash n} \max_{z\in\D(\la/n)}
\exp\big( y\cdot\la+ n\, \phi(z)+n\d_n\big)\\
&\leq n^{2\theta} \cdot \max_{x\in\D} \max_{z\in\D(x)}
\exp\big( ny\cdot x+ n\, \phi(z)+n\d_n\big)\\
&\leq n^{2\theta} \cdot 
\exp\big( \max_{x\in\D} \big\{ny\cdot x+n g(x)\big\}+n\d_n\big)
\end{split}
\]
Thus
\[
\frac1n\log \Big(
\sum_{\la\vdash n} e^{y\cdot \la}
\sum_{\mu\dom\la}\exp\big( n\, \phi^{(\la)}_n(\mu)\big)
\Big)\leq m(y)+o(1).
\]
For the lower bound, note that given $\d>0$ we
 may find $\tilde x\in\D$ such that
\[
y\cdot \tilde x+g(\tilde x) \geq m(y)-\d.
\]
We may also find a sequence $\tilde\la\vdash n$ such that
$\tilde\la/n\to \tilde x$.  Clearly
\[
\sum_{\la\vdash n} e^{y\cdot\la}
\sum_{\mu\dom\la}\exp\big( n\, \phi^{(\la)}_n(\mu)\big)
\geq e^{y\cdot\tilde \la}
\sum_{\mu\dom\tilde\la}
\exp\big( n\, \phi^{(\tilde\la)}_n(\mu)\big)
\]
and hence
\[
\frac1n\log \Big(
\sum_{\la\vdash n} e^{y\cdot\la}
\sum_{\mu\dom\la}\exp\big( n\, \phi^{(\la)}_n(\mu)\big)
\Big)\geq 
\frac{y\cdot \tilde\la}{n}+
\frac1n\log \Big(
\sum_{\mu\dom\tilde\la}\exp\big( n\, \phi^{(\tilde\la)}_n(\mu)\big)
\Big).
\]
By Lemma~\ref{penrose-lem} the right-hand-side converges to 
\[
y\cdot \tilde x+g(\tilde x) \geq m(y)-\d.
\]
This proves the claim.
\end{proof}

\subsection{The free energy}

From now on we let $\phi=\phi_\b:\D\to\RR$ be the function given
in~\eqref{phi-def}, i.e.\ 
$\phi_\b(x)=\frac\b2\big(\sum_{i=1}^\theta x_i^2-1\big)
-\sum_{i=1}^\theta x_i \log x_i$.
Note that $\phi_\b$ is continuous.
We write $g_\b(x)=\max_{y\dom x}\phi_\b(y)$ and
we define
\be\label{zbh-def}
z(\b,h)=\left\{
\begin{array}{ll}
\max_{x\in\D}\big(h(x_1-\tfrac1\theta)+g_\b(x)\big),
& \mbox{if } h\geq 0,\\
\max_{x\in\D}\big(h(x_\theta-\tfrac1\theta)+g_\b(x)\big),
& \mbox{if } h\leq 0.
\end{array}\right.
\ee
Note that $x_1-\tfrac1\theta\geq0$ and
$x_\theta-\tfrac1\theta\leq0$.

The following theorem
contains Theorem~\ref{free-en-thm} as the case $h=0$.
\begin{theorem}\label{free-energy-thm}
We have that
\[
\tfrac1n\log G_n(\la)\to g_\b(x),
\mbox{ as }n\to\oo\mbox{ and } \la/n\to x,
\]
and for $h\in\RR$ that
\[
\tfrac1n\log Z_n(\b,h)\to z(\b,h),
\mbox{ as }n\to\oo.
\]
\end{theorem}
\begin{proof}
We will use  Lemmas~\ref{penrose-lem} and~\ref{penrose-lem-2} with  
\be
\phi^{(\la)}_n(\mu)=\frac{\b}{n^2}\binom{n}{2}
[r(\mu)-1]+\tfrac1n \log d_\mu
+\tfrac1n \log K_{\mu\la}.
\ee
Due to Lemmas~\ref{pf-lem} and~\ref{G-lem}
it suffices to establish the uniform convergence of
$\phi^{(\la)}_n(\mu)$ to $\phi=\phi_\b$.  First note that 
$K_{\mu\la}\leq (n+1)^{\theta^2}$.  Indeed, for each row of $\mu$ we must
choose the number of 1's, the number of 2's etc.  Thus there are
certainly at most $(\la_1+1)\dotsb(\la_\theta+1)$ choices for each row,
and thus
\[
K_{\mu\la}\leq [(\la_1+1)\dotsb(\la_\theta+1)]^\theta\leq (n+1)^{\theta^2},
\]
as claimed.  Defining
\[
\phi_n(\mu)=\frac{\b}{n^2}\binom{n}{2}
[r(\mu)-1]+\tfrac1n \log d_\mu
\]
we thus have that
\[
|\phi^{(\la)}_n(\mu)-\phi(\mu/n)|\leq
 |\phi_n(\mu)-\phi(\mu/n)|
+\tfrac{\theta^2}{n}\log (n+1).
\]
Now by~\eqref{r-eq}  we have
\[
\frac{\b}{n^2}\binom{n}{2}[r(\mu)-1]
=\frac\b2\Big[
\sum_{j=1}^\theta \frac{\mu_j(\mu_j-2j+1)}{n^2}
-\frac{n-1}{n}\Big]
\]
and we have that
\[
\Big|
\sum_{j=1}^\theta \frac{\mu_j(\mu_j-2j+1)}{n^2}
-\sum_{j=1}^\theta\big(\frac{\mu_j}{n}\big)^2\Big|
\leq \sum_{j=1}^\theta \frac{\mu_j}{n}
\big(\frac{2j-1}{n}\big)\leq \frac{2\theta-1}{n}.
\]
Next,  (4.11) on page 50 of~\cite{fulton-harris} 
gives that 
\[
\log d_\mu=\log\Big(\frac{n!}{m_1!\dotsb m_k!} 
\prod_{1\leq i<j\leq k}(m_i-m_j)\Big)
\]
where $m_i=\mu_i+k-i$ and $k$ is the number of 
\emph{nonzero} parts of
$\mu$.  Thus
\[\begin{split}
&\Big|\frac1n\log d_\mu-\frac1n\log\binom n\mu\Big|\leq
\frac1n\log \prod_{1\leq i<j\leq k}(m_i-m_j)\\
&+\frac1n\log [(\mu_1+k-1)\dotsb(\mu_1+1)
  (\mu_2+k-2)\dotsb (\mu_2+1)\dotsb (\mu_{k-1}+1)]\\
&\leq \frac1n\log (n+\theta-1)^{\theta^2}
+\frac1n\log (n+\theta-1)^\theta.
\end{split}\]
Thus it suffices to bound
\[
\Big|\frac1n\log\binom n\mu-
\Big(-\sum_{j=1}^\theta \frac{\mu_j}{n}\log\frac{\mu_j}{n}\Big)\Big|.
\]
But by Stirling's formula
\[
\binom{n}{\mu}\asymp 
\Big(\frac{n}{\prod_{j=1}^\theta \mu_j}\Big)^{1/2}
\prod_{j=1}^\theta \Big(\frac{n}{\mu_j}\Big)^{\mu_j}
\]
so that
\[
\Big|\frac1n\log\binom n\mu-
\Big(-\sum_{j=1}^\theta \frac{\mu_j}{n}\log\frac{\mu_j}{n}\Big)\Big|
\leq \Big|\frac1n \log 
\Big(\frac{n}{\prod_{j=1}^\theta \mu_j}\Big)^{1/2}\Big|
+\frac Cn.
\]
The right-hand-side is at most $\tfrac{C'}{n} \log n$.
This proves the result.
\end{proof}

\section{The phase-transition}
\label{pt-sec}

In this last section we prove Theorem~\ref{mag-thm} and 
Proposition~\ref{X-prop}.
Recall  $\phi_\b$ and $z(\b,h)$
defined in~\eqref{phi-def} and~\eqref{zbh-def},
respectively.

\subsection{Left and right derivatives of $z(\b,h)$ at $h=0$}

Let $x^\uparrow(\b)\in\D$ denote a maximizer of $\phi_\b$ 
which maximizes the
first coordinate.  That is, among the maximizers $x$ of
$\phi_\b$ we pick one for which $x_1$ is maximal.  
Similarly, let $x^\downarrow(\b)$ denote a maximizer 
of $\phi_\b$ which
\emph{minimizes} the last coordinate $x_\theta$.
Note that $x^\uparrow$ and $x^\downarrow$ depend on
$\b$, though we do not always write this explicitly.

The left and right derivatives of $z(\b,h)$ at $h=0$
are given by
\[
z^+(\b)=\lim_{h\downarrow0}\frac{z(\b,h)-z(\b,0)}{h},\qquad
z^-(\b)=\lim_{h\uparrow0}\frac{z(\b,h)-z(0)}{h}.
\]
We will show:
\begin{theorem}\label{der-thm}
\[
z^+(\b)=x_1^\uparrow(\b)-\tfrac1\theta,\qquad
z^-(\b)=x_\theta^\downarrow(\b)-\tfrac1\theta.
\]
\end{theorem}
\begin{proof}
We prove the  claim about $z^+(\b)$;  
the argument for $z^-(\b)$ is similar.
First note that $z(\b,0)=\phi_\b(x^\uparrow)$ and so
\[
\frac{z(\b,h)-z(\b,0)}{h}=\max_{x\in\D} f(x,h), 
\]
where 
\[
f(x,h)=x_1-\tfrac1\theta +\frac{g_\b(x)-\phi_\b(x^\uparrow)}{h}.
\]
We have that $f(x^\uparrow,h)=x_1^\uparrow-\tfrac1\theta$, 
since $g_\b(x^\uparrow)=\phi_\b(x^\uparrow)$, and
thus 
\[
\frac{z(\b,h)-z(\b,0)}{h}\geq x_1^\uparrow-\tfrac1\theta
\mbox{ for all }h>0.
\]
Also, $f$ is continuous as a function on $\D\times(0,\oo)$, thus for
each $h$ it attains its maximum at some point $x(h)\in\D$.
Since $g_\b(x)\leq\phi_\b(x^\uparrow)$ for all $x\in\D$ it follows that 
\[
x_1^\uparrow-\tfrac1\theta\leq f(x(h),h)\leq x_1(h)-\tfrac1\theta,
\mbox{ for all }h>0.
\]
It thus suffices to show 
that $x_1(h)\to x_1^\uparrow$ as $h\downarrow0$.

If not then there is some $\eps>0$ and some sequence
$h_i\downarrow0$ such that $x(h_i)\in A_\eps$ for all $i$, where
\[
A_\eps=\{x\in\D:x_1\geq x_1^\uparrow+\eps\}.
\]
Note that there is some $\d>0$ such that 
$\phi_\b(x)\leq\phi_\b(x^\uparrow)-\d$ for all $x\in A_\eps$, 
since $\phi_\b$ is continuous and $A_\eps$ compact.
Also note that if $x\in A_\eps$ then $\D(x)\se A_\eps$, by
the definition of $\dom$.  
Thus $g_\b(x(h_i))\leq\phi_\b(x^\uparrow)-\d$ for
all $i$.  But then
\[
f(x(h_i),h_i)= x_1(h_i)-\tfrac1\theta
+\frac{g_\b(x(h_i))-\phi_\b(x^\uparrow)}{h_i}
\leq 1-\tfrac1\theta-\frac{\d}{h_i}\to -\oo.
\]
This contradicts the fact that 
$f(x,h)\geq x^\uparrow_1-\tfrac1\theta$ for all $x\in\D$.
Hence it must be the case that $x_1(h)\to x_1^\uparrow$, as claimed.
\end{proof}

\subsection{The critical point}

In light of Theorem~\ref{der-thm}, the following result 
implies Theorem~\ref{mag-thm}.
Recall that 
$\b_\crit(\theta):=
2\big(\frac{\theta-1}{\theta-2}\big)\log(\theta-1)$
for $\theta\geq3$
and $\b_\crit(2)=2$.

\begin{theorem}\label{critval-thm}
\hspace{1cm}\\
If $\b<\b_\crit$, or $\theta=2$ and $\b=\b_\crit$, 
then $x^\uparrow_1=x^\downarrow_\theta=\tfrac1\theta$.

\noindent
If $\b>\b_\crit$, or $\theta\geq3$ and $\b=\b_\crit$, 
then $x^\uparrow_1>\tfrac1\theta$
and $x^\downarrow_\theta<\tfrac1\theta$.
\end{theorem}
\begin{proof}
Since $x_1\geq\tfrac1\theta$ and $x_\theta\leq\tfrac1\theta$
for all $x\in\D$ we must determine
when $\phi_\b$ has a maximizer different from 
$(\tfrac1\theta,\dotsc,\tfrac1\theta)$.
We start by characterizing the possible
maxima of $\phi_\b$ using the Lagrangian
necessity theorem.  Since the 
functions $\phi_\b(x)$
and $c(x)=\sum x_i-1$ are  $C^1$ on $(0,\oo)^\theta$, and 
$\nabla c(x)$ is nonzero for all $x$, if $x\in\D$ is any local
extremum  of $\phi_\b$ then there is some $a\in\RR$ such that
\[
\nabla\phi_\b(x)=a\nabla c(x)=(a,\dotsc,a).
\]
Now
\be\label{partial-phi}
\frac{\partial\phi_\b}{\partial x_i}=\b x_i-\log x_i-1
\ee
so if $x\in\D$ is a local maximum then there is some $a\in\RR$ such
that 
\be\label{loc-fp}
\b x_i=(1-a)+\log x_i,\qquad\mbox{ for all } i=1,\dotsc,\theta.
\ee
(We see from~\eqref{partial-phi} that the partial derivative diverges
to $+\oo$ if $x_i\downarrow0$, 
thus $\phi_\b$ is not maximized on the boundary
and it suffices to consider local maxima.)
For each $\b>0$ and $a\in\RR$, there are 0, 1 or 2 values of $x_i$
which satisfy~\eqref{loc-fp}.  If there is just 1 solution then all
the $x_i$ are equal, and hence equal to $\tfrac1\theta$.
If there are 2 solutions then, since $\phi_\b$ is symmetric in its
arguments, we can assume that
there is some $1\leq r\leq \theta-1$ such that 
$x$ is of the form
\be\label{x-r}
x=(t,\dotsc,t,\tfrac{1-rt}{\theta-r}, \tfrac{1-rt}{\theta-r})
\qquad
\mbox{for } \tfrac1\theta< t < \tfrac1r, 
\ee
with the first $r$ coordinates equal and the last $\theta-r$
coordinates equal.  Write $\phi^{(r)}_\b(t)$ 
for $\phi_\b$ evaluated at $x$ of
the form~\eqref{x-r}.

We now establish a condition on $\b$ for $\phi^{(r)}_\b(t)$ to exceed
$\phi^{(r)}_\b(\tfrac1\theta)$ for some $t>\tfrac1\theta$.  A short
calculation shows that 
\[
\phi^{(r)}_\b(t)-\phi^{(r)}_\b(\tfrac1\theta)=\frac{\b r}{2\theta(\theta-r)}
(\theta t -1)^2
-[rt\log t+(1-rt)\log\tfrac{1-rt}{\theta-r}+\log\theta].
\] 
Thus $\phi^{(r)}_\b(t)-\phi_\b^{(r)}(\tfrac1\theta)\geq 0$ if and only if
\be\label{R-eq}
\b\geq R(t)=\Big(\frac{2\theta(\theta-r)}{r}\Big)
\frac{rt\log t+(1-rt)\log\tfrac{1-rt}{\theta-r}+\log\theta}
{(\theta t-1)^2}.
\ee
 Hence $\phi_\b$ has a maximizer different 
from $(\tfrac1\theta,\dotsc,\tfrac1\theta)$
if and only if $\b\geq R(t)$ for some $r$ and some $t>\tfrac1\theta$.
We show in Appendix~\ref{R-app} that $R$ is convex.  Also, 
note that $R(t)\to+\oo$ as $t\uparrow\tfrac1r$ and
that $R'(\tfrac{\theta-r}{r\theta})=0$.  Thus $R(t)$ has a unique
minimum in $[\tfrac1\theta,\tfrac1r)$, 
either at the boundary point $t=\tfrac1\theta$ if $r>\theta/2$,
or at $t=\tfrac{\theta-r}{r\theta}$ if $r\leq\theta/2$.  

In the case when $\theta=2$ the only possibility for $t>\tfrac1\theta$
is when $r=1$.  Then 
$\tfrac1\theta=\tfrac{\theta-r}{r\theta}=\tfrac12$ and
hence $\b_\crit(2)=\inf_{t>1/2}R(t)=2$.
If $\theta\geq3$, note that
\be\label{beta-f}
R(\tfrac{\theta-r}{r\theta})=\rho(\tfrac r\theta),
\mbox{ with }
\rho(t)=2\theta t \frac{1-t}{1-2t}
\log\Big(\frac{1-t}{t}\Big).
\ee 
The function $\rho$ is increasing on $[0,\tfrac12]$, 
so $\rho(\tfrac r\theta)$ is minimal for $r=1$.
This gives the critical value $\b_\crit=\rho(\tfrac1\theta)$
claimed.

To check the statements about
$x^\uparrow$ and $x^\downarrow$
at $\b=\b_\crit$, we note that for this
value of $\b$ we have a maximizer of $\phi_\b$ at the point~\eqref{x-r}
with $r=1$ and $t=\tfrac{\theta-1}{\theta}$.
Thus, at $\b=\b_\crit$,
\[
x^\uparrow_1= \tfrac{\theta-1}{\theta}\quad\mbox{and}
\quad
x^\downarrow_\theta= \tfrac{1}{(\theta-1)\theta}.
\]
The claims follow.
\end{proof}

\subsection{The number of vertices in large cycles}

Let $k=k_n\to\oo$ be any sequence going to
$\oo$.  Recall that $X_n(k)=\frac1n\sum_{|\g|\geq k} |\g|$
denotes the fraction of vertices in cycles of
size at least $k$ in the random permutation $\s(\om)$.
We now show that, under $\PP_\theta$ with
$\theta\in\{2,3,\dotsc\}$, asymptotically $X_n(k)$ is at most 
\be
\frac{\theta x_1^\uparrow-1}{\theta-1}.
\ee
Note that this number is $=0$ if and only 
if $x^\uparrow_1=\tfrac1\theta$,
i.e.\ $z^+(\b)=0$.  Proposition~\ref{X-prop}
is a special case of the following result:

\begin{proposition}\label{expdecay-prop}
Let $\b>0$.
For any $\a<1-\tfrac1\theta$ 
and any $\eps>0$ there is some $c>0$ such that
\be\label{expdecay}
\PP_\theta\big(X_n(k)>\eps+\tfrac1\a(x_1^\uparrow-\tfrac1\theta)\big)
\leq e^{-c n}
\ee
for all large enough $n$.
\end{proposition}
\begin{proof}
We claim that for any
$h>0$ we have for large enough $n$ that
\be\label{exp-bd}
\EE_\theta\Big[
\exp\Big(\a h\sum_{|\g|\geq k} |\g|\Big)
\Big]\leq
\frac{Z_n(\b,h)}{Z_n(\b,0)}.
\ee
Indeed,
\be\label{prod-split}
\begin{split}
Z_n(\b,h)&=\EE\Big[
\prod_{\g}\Big(\frac{e^{h |\g|}+\theta-1}{e^{(h/\theta)|\g|}}
\Big)\Big]\\
&=\EE\Big[\theta^\ell
\prod_{|\g|\geq k}
w(h|\g|)
\prod_{|\g|<k}
w(h|\g|)
\Big)\Big],
\end{split}
\ee
where
\[
w(x)=\frac{e^{x}+\theta-1}{\theta e^{x/\theta}}
=\tfrac1\theta e^{x(1-\sfrac1\theta)}+
\tfrac{\theta-1}{\theta}e^{-\sfrac{x}{\theta}}
\]
is increasing in $x\geq0$, and satisfies:
\be
w(x)\geq
\left\{\begin{array}{ll}
w(0)=1, & \mbox{for all }x\geq0, \\
e^{\a x}, & \mbox{for all large enough }x. 
\end{array}\right.
\ee
It follows from~\eqref{prod-split} that for large enough $n$,
\be
Z_n(\b,h)\geq
\EE\Big[\theta^\ell
\exp\Big(\a h\sum_{|\g|\geq k} |\g|\Big)
\Big],
\ee
which gives the claim.

For any $\eps>0$ we have that
\[
\PP_\theta\big(X_n(k)>\eps+\tfrac1\a(x_1^\uparrow-\tfrac1\theta)\big)=
\PP_\theta\Big(\a h \sum_{|\g|\geq k}|\g|>
hn(\a\eps +x_1^\uparrow-\tfrac1\theta)\Big).
\]
Using Markov's inequality and~\eqref{exp-bd}
it follows that
\[
\PP_\theta\big(X_n(k)>\eps+\tfrac1\a(x_1^\uparrow-\tfrac1\theta)\big)
\leq 
\exp(-hn(\a\eps+x^\uparrow_1-\tfrac1\theta)) \frac{Z_n(\b,h)}{Z_n(\b,0)}.
\]
Thus
\be
\begin{split}
\limsup_{n\to\oo}\frac1n &
\log \PP_\theta\big(X_n(\om)>\eps+\tfrac1\a(x_1^\uparrow-\tfrac1\theta)\big)
\\
&\leq -h(\a\eps+x^\uparrow_1-\tfrac1\theta)+z(\b,h)-z(\b,0)\\
&= h\Big(\frac{z(\b,h)-z(\b,0)}{h}-\a\eps-(x^\uparrow_1-\tfrac1\theta)\Big).
\end{split}
\ee
By Theorem~\ref{der-thm} we have  that
$\lim_{h\downarrow 0} \frac{z(\b,h)-z(\b,0)}{h} =x_1^\uparrow-\tfrac1\theta$,
hence there is some $h>0$ such that
\be
\limsup_{n\to\oo}\frac1n 
\log \PP_\theta\big(X_n(\om)>\eps+\tfrac1\a(x_1^\uparrow-\tfrac1\theta)\big)
\leq-\frac{h\a\eps}{2}.
\ee
This proves the  result.
\end{proof}

\appendix

\section{The transposition-operator}
\label{T-app}

Let $T_{xy}$ be the transposition operator~\eqref{T-def}
acting on the tensor product $\cH_x\otimes\cH_y$ of two copies of
$\cH=\CC^{2S+1}$.
Write $\ket{a,b}$ for the `uncoupled' basis
$\ket{a,b}=\ket a \otimes \ket b$
of $\cH_x\otimes\cH_y$ so that 
$T_{xy}\ket{a,b}=\ket{b,a}$.
We wish to express $T_{xy}$ in terms of the spin operators 
$\bS_x\cdot\bS_y$.
We will show:
\begin{proposition}\label{T-prop}
For each $S\in\tfrac12\NN$ there are 
$a_0,a_1,\dotsc,a_{2S}\in\QQ$ such that
\[
T_{xy}=\sum_{k=0}^{2S} a_k (\bS_x\cdot\bS_y)^k.
\]
\end{proposition}
\begin{proof}
We will use standard properties of additions of spins, in particular
the operator 
\be\label{J-def}
\bJ=(\bS_x+\bS_y)^2=2 S(S+1)I+2\bS_x\cdot\bS_y.
\ee
See e.g.~\cite[Chapter~V]{bohm} for background.
There is an orthonormal basis for $\cH_x\otimes\cH_y$ consisting 
of eigenvectors of $\bJ$, which we denote
\[
\pet{J,M}\mbox{ for }J\in\{0,1\dotsc,2S\}
\mbox{ and } M\in\{-J,\dotsc,J\}.
\]
Note that we use the notation $\pet{J,M}$ for this basis, and
$\ket{a,b}$ for the uncoupled basis.
We have
\be\label{J-eig}
\bJ\pet{J,M}= J(J+1)\pet{J,M}.
\ee
The basis-change matrix from the basis $\pet{J,M}$ 
to the uncoupled basis $\ket{a,b}$ 
is given by the Clebsch--Gordan coefficients
$(J,M\ket{a,b}\in\RR$:
\[
\ket{a,b}=\sum_{J,M} \pet{J,M} (J,M\ket{a,b}.
\]
These coefficients satisfy the relation:
\be\label{JM-ba}
(J,M\ket{a,b}=(-1)^{2S-J}(J,M\ket{b,a}.
\ee
Write $\bK$ for an operator
$\bK=\sum_{k=0}^{2S} a_k(\bS_x\cdot\bS_y)^k$,
where the coefficients $a_k$ are to be chosen.  Using~\eqref{J-def}
and~\eqref{J-eig} we see that
\[
(\bS_x\cdot\bS_y)^k\pet{J,M}=x_J^k\pet{J,M},
\mbox{ where }
x_J=\tfrac12J(J+1)-S(S+1).
\]
We claim that we can pick the coefficients $a_k\in\QQ$ so that 
$\bK\pet{J,M}=(-1)^{2S-J}$ for all $J$.  Indeed, this holds if and
only if the following matrix-equation holds:
\[
\begin{pmatrix}
1 & x_0 & x_0^2 & \dotsb & x_0^{2S} \\
1 & x_1 & x_1^2 & \dotsb & x_1^{2S} \\
 & \vdots &  &  &   \\
1 & x_{2S} & x_{2S}^2 & \dotsb & x_{2S}^{2S} 
\end{pmatrix}
\begin{pmatrix}
 a_0 \\
a_1 \\
\vdots\\
a_{2S}
\end{pmatrix}
=
\begin{pmatrix}
 (-1)^{2S} \\
(-1)^{2S-1} \\
\vdots\\
(-1)^{2S-2S}
\end{pmatrix}.
\]
The Vandermonde-matrix on the left is invertible and has rational
entries, thus its inverse has rational entries and the claim follows. 

With this choice of the $a_k$ we have, using~\eqref{JM-ba},
\[
\begin{split}
\bK\ket{a,b}&=\sum_{J,M}\bK\pet{J,M} (J,M\ket{a,b}
=\sum_{J,M} (-1)^{2S-J}\pet{J,M} (J,M\ket{a,b}\\
&=\ket{b,a},
\end{split}\]
thus $\bK=T_{xy}$ as required.
\end{proof}

\section{Proof of Lemma~\ref{D-lem}}
\label{dh-app}

Recall that
we need to show that if $x,y\in\D$ with $\|x-y\|\leq\eps<\theta^{-2}$ 
then
$d_{\mathrm{H}}(\D(x),\D(y))<\theta\eps^{1/2}$.
Take an arbitrary $z\in\D(x)$.  
We will show that there is some $z'\in\D(y)$
satisfying $\|z-z'\|\leq(\theta-1)\eps^{1/2}$.  
This suffices, by the
symmetry between $x$ and $y$.  

Define
$k=\max\{j\geq1: z_j\geq \eps^{1/2}\}$.
 Since $z_1\geq\tfrac1\theta>\eps^{1/2}$, we have
that $k\geq1$.  Also let $\a=(k-1)\eps^{1/2}\in[0,1]$.
We claim that the following choice of $z'$ satisfies
our requirements:
\[
\begin{split}
z_1'&=\a +(1-\a)z_1+(1-\a)\tsum_{i=k+1}^\theta z_i,\\ 
z_i'&=(1-\a)z_i,\quad\mbox{ for }2\leq i\leq k,\\
z_i'&=0,\quad\mbox{ for }i>k.
\end{split}
\]
First we check that $\|z-z'\|\leq(\theta-1)\eps^{1/2}$.
Indeed, we have that
$|z_1-z_1'|\leq\a+(\theta-k)\eps^{1/2}=(\theta-1)\eps^{1/2}$,
that $|z_i-z_i'|\leq\a=(k-1)\eps^{1/2}$ for $2\leq i\leq k$,
and that $|z_i-z_i'|=z_i<\eps^{1/2}$ for $i>k$.

Next we check that $z'\dom y$.  If $j\geq k$ then
clearly
$\tsum_{i=1}^j z_i'=1\geq \tsum_{i=1}^j y_i$.
Now let $j<k$.  Firstly, since
$z\dom x$ and $\|x-y\|\leq\eps$, we have that
\[
\tsum_{i=1}^j z_i-\tsum_{i=1}^j y_i\geq
\tsum_{i=1}^j (x_i-y_i)\geq -j\eps.
\]
Hence we see that
\[
\begin{split}
\tsum_{i=1}^j z_i'-\tsum_{i=1}^j y_i&\geq
\a +\big(\tsum_{i=1}^j z_i-\tsum_{i=1}^j y_i\big)
-\a \tsum_{i=1}^j z_i\\
&\geq -j\eps + \a\tsum_{i=j+1}^k z_i\\
&\geq -j\eps+(k-1)(k-j)\eps\geq 0.
\end{split}
\]
The result follows.
\qed

\section{Convexity of the function $R(t)$}
\label{R-app}

Recall that we needed to know that a certain 
function $R(t)$, given in~\eqref{R-eq},
is convex.
The function $R(t)$ is (up to a constant factor)
\[
R(t)=
\frac{rt\log t+(1-rt)\log\tfrac{1-rt}{\theta-r}+\log\theta}
{(\theta t-1)^2},\qquad
\tfrac1\theta<t<\tfrac1r.
\]
Actually $R(t)$ is well-defined for all $t\in(0,1/r)$.  With $s=rt$,
$p=\theta/r$ and $q=\theta/(\theta-r)$ the convexity of $R(t)$ follows
from the following result.

\begin{lemma}
For any $p,q>0$ satisfying $\tfrac1p+\tfrac1q=1$, the function
\[
f(s)=\frac{s\log(ps)+(1-s)\log(q(1-s))}{(s-\tfrac1p)^2},
\quad s\in(0,1)
\]
is convex.
\end{lemma}
\begin{proof}
Let $g(s)=(s-\tfrac1p)^4f''(s)$.  Direct computation gives
\[
g(s)=\frac{(s-\tfrac1p)^2}{s(1-s)}+
(2s+\tfrac4p)\log(ps)+
(6-(2s+\tfrac4p))\log(q(1-s))
\]
and thus
\[
g''(s)=\frac{2(s-\tfrac1p)^2(1-3s+3s^2)}{s^3(1-s)^3}\geq0.
\]
Hence $g$ is convex on $(0,1)$.  It is easy to see that $g(s)\to+\oo$
as $s\downarrow0$ or $s\uparrow1$, so $g$ has a unique 
minimum in $(0,1)$.  We have that
\[\begin{split}
g'(s)&=2\log(ps)-2\log(q(1-s))\\
&\quad+
\frac{\tfrac4p+2s}{s}+\frac{(s-\tfrac1p)^2}{s(1-s)^2}
-\frac{(s-\tfrac1p)^2}{s^2(1-s)}\\
&\quad+2\frac{s-\tfrac1p}{s(1-s)}
-\frac{6-\tfrac4p-2s}{1-s},
\end{split}\]
and therefore $g'(\tfrac1p)=0$.  It follows that 
\[
g(s)\geq g(\tfrac1p)=0\mbox{ for all }s\in(0,1).
\]
We conclude  that $f''(s)\geq0$ for all $s\in(0,1)$, as required.
\end{proof}

\subsection*{Acknowledgement}
Part of this work was carried out while the author was at
Chalmers University of Technology in G\"oteborg, Sweden, with
support from the Knut and Alice Wallenberg Foundation.

\end{document}